\documentclass[runningheads]{llncs}
\usepackage{amsmath} 
\usepackage{amssymb}
\usepackage{comment}
\usepackage{xspace}
\usepackage{subcaption} 
\usepackage{graphicx}
\usepackage{tikz}
\usetikzlibrary{math, calc, backgrounds}
\usepackage{listings}
\usepackage{xcolor}
\usepackage{hyperref}
\usepackage{bbding}
\usepackage{todonotes}
\newcommand{\noI}{\ensuremath{\emptyset}\xspace}
\begin{document}
\title{Onto-DP: Constructing Neighborhoods for Differential Privacy on Ontological Databases}
\titlerunning{Onto-DP: DP on Ontological Databases}
%
\author{Yasmine Hayder \Envelope \and Adrien Boiret \and Cédric Eichler\orcidID{0000-0003-3026-1749} \and Benjamin Nguyen}
\authorrunning{Y. Hayder et al.}
%
\institute{LIFO, INSA CVL, Univ. Orléans, Inria, France\\
\email{firstname.lastname@insa-cvl.fr}\\
}
\maketitle              
\begin{abstract}
In this paper, we investigate how attackers can discover sensitive information embedded within databases by exploiting inference rules. We demonstrate the inadequacy of naively applied existing state of the art differential privacy (DP) models in safeguarding against such attacks. 

We introduce ontology aware differential privacy (Onto-DP), a novel extension of differential privacy paradigms built on top of any classical DP model by enriching it with semantic awareness. We show that this extension is a sufficient condition to adequately protect against attackers aware of inference rules.

\keywords{Privacy \and Differential Privacy \and Inferences \and Ontology.}\end{abstract}

\section{Introduction} 
Databases in general, and semantic databases such as Knowledge Graphs (KGs) in particular, are often used to store personal and/or private information, such as healthcare data~\cite{abu2023healthcare}. In this article, we are interested in the \textit{privacy} issues linked to protecting the personal information in databases when publishing the results of queries. These privacy questions have been a pressing issue since seminal works of Sweeney~\cite{DBLP:journals/ijufks/Sweene02} on $k$-anonymity, and a whole field called \textit{differential privacy} (DP)~\cite{dwork2006differential} has seen a very strong development these last twenty years. DP has garnered great interest within both theoretical and applied database and privacy communities. The general idea behind DP is to bound the relative information gain on the underlying data that a querier obtains when observing the results of their query. One of the most important features of DP is that the approach (supposedly) makes no assumption on any background knowledge that the attacker may possess, which was an important breakthrough, compared to previous works such as $k$-anonymity, where a large part of the practical difficulties when trying to evaluate the security of the approach is that it depends on the attacker's background knowledge. However, there is in fact an implicit dependency of DP on some kind of attacker background knowledge: DP is based on the concept of \textit{adjacent} databases, loosely defined as two databases that differ by ``one element''.  In this article, we consider that the adjacent databases of a given database $D$ are given by defining a distance $d$ and its associated \textit{neighborhood} $N_d(D)$. Once this neighborhood definition is given, it is then possible to directly apply state-of-the-art DP mechanisms to protect the data, such as adding random noise drawn from a \textit{Laplace distribution}. 

A lot of research in DP has gone into building optimal mechanisms, however we focus on a much less studied question: \textit{how to correctly define this distance}, in order to capture what we call \textit{semantic attackers} that possess knowledge about inference rules (data dependency). To our knowledge, we are the first to (correctly) build such a distance. The main difficulty of this research question is that classical distances used in DP are expressed over a representation-dependent notion of single datum (e.g. relational databases differ by a single tuple, graph databases differ by a single node or edge, RDF databases differ by a single triple, etc.), but this may not precisely translate to two databases differing in one fact, in the semantic sense, since data items are often correlated, especially in graph structured databases. Works such as Pufferfish~\cite{kifer2014pufferfish} provide a very wide setting to customize privacy for dependent data by defining secrets, alternative worlds that must be indistinguishable, and an attacker’s background knowledge. However, while DP and deterministic inference considerations can be expressed in this setting
with a prohibitively comprehensive class of attackers,
in practice showing that a process respects a privacy constraint or not
is best checked on a small set of worst-case-scenario attackers.
No such suitable definition that we know of exists for databases with inference.

\textbf{To address these limitations}, we (i) formalize both the concepts of a \textit{semantically aware attacker}, \textit{i.e.} one that knows how to infer facts from existing information and knows all but one information, and of a distance providing a \textit{well suited defense} w.r.t. a class of attackers, in the context of a DP-based protection; (ii) demonstrate that classical DP approaches on databases, while \textit{well suited} w.r.t. semantic-unaware attackers, are \textit{ill suited} w.r.t. attackers aware of inference rules; (iii) show that such ill-suitedness can lead to incorrect estimation of privacy and leaks of supposedly protected data; (iv) propose \textit{onto-DP}, a consistent extension built on top of existing DP models, and demonstrate that it is (by construction) \textit{well-suited} for semantic-aware attackers.

The remainder of this paper is structured as follows. Sec.~\ref{sec:bg} introduces the background on DP and (semantic) databases. 
The problem we study is presented and formalized in Sec.~\ref{sec:pbstate}. The incorporation of inference rules and the possible mismatch is discussed in Sec.~\ref{sec:fullysemantic}. The proposed distance that is well-suited to semantic-aware attackers is presented in Sec.~\ref{sec:solution}. The related work and positioning of our paper are presented in Sec.~\ref{sec:rw}. Finally, Sec.~\ref{sec:ccl} concludes and presents future research directions.

\section{Background and notations}
\label{sec:bg}
In this section, we provide the background on the data model, inference rules, differential privacy, and distances necessary for the theoretical foundations of this paper. \\

\noindent\textbf{Private database $D$.} We consider a database $D$ containing (sensitive) information to be queried. We do not make any extra hypothesis on the structure of $D$: it may be relational, a KG, etc.
We note $\mathcal{D}$ the set of databases considered and $D \subseteq D'$ the fact that $D$ is contained in $D'$. As we consider attackers that conduct inferences, we will instantiate $D$ as a KG in our examples, but our results hold in the general case. \\

\noindent\textbf{Knowledge graphs} are structured representations in graphs that model real-world entities as nodes and their relationships as edges, or (subject-predicate-object) triples. RDF and Neo4j~\cite{vukotic2014neo4j} are common frameworks to represent KGs.

\begin{example}[Hospital DB]\label{ex:KG} A toy example knowledge graph database
\label{ex:hospital-db}
\begin{lstlisting}[frame=single]
# Example Hospital Database
@prefix ex:   <http://example.com/hospital#> .
ex:d1	a	ex:doctor ;
	ex:worksIn ex:dept1 ;
	ex:hasPatient ex:p1 .
ex:dpt1	a	ex:dept .
ex:p1	a 	ex:patient .
\end{lstlisting}
\end{example}

\noindent\textbf{Inference Rules $I$} are mechanisms that allow the derivation of new knowledge from existing facts by applying logical reasoning to help enrich the database with new facts. They may be expressed in various knowledge representation languages, from the simplest ones like RDFS~\cite{rdfs} to much more expressive languages such as SWRL~\cite{horrocks2004swrl}. 
Reasoner engines such as Hermit~\cite{shearer2008hermit} 
are used to derive new facts in the database, based on its contents and the rules considered. 

\begin{example}[Inferring new tuples] \label{ex:saturated-hospital-db} Consider an inference rule stating that a patient under the care of a physician working in a particular department is a patient in said department, written in SWRL "\textit{human readable syntax}" format~\cite{horrocks2004swrl}:
\[
IR = \{\texttt{hasPatient(?x, ?y)} \land  \texttt{worksIn(?x, ?z)} \Rightarrow \texttt{patientIn(?y, ?z)} \}
\]
We show next the hospital database of Example~\ref{ex:hospital-db}, after applying a reasoner.
\begin{lstlisting}[frame=single]
# Example Hospital Database after reasoning
@prefix ex:   <http://example.com/hospital#> .
ex:d1	a	ex:doctor ;
	ex:worksIn ex:dept1 ;
	ex:hasPatient ex:p1 .
ex:dpt1	a	ex:dept .	
ex:p1	a 	ex:patient .
ex:p1	ex:patientIn ex:dept1 .
\end{lstlisting}

\end{example}
We formally define the inference system as follows:
\begin{definition}[Inference System]
    Let $\mathcal D$ be a space of databases, and $I$ a set of inference rules. An inference system $I:\mathcal{D}\rightarrow \mathcal{D}$ is a function that associates some database $D$ of $\mathcal{D}$ to its \emph{saturated} version $I(D)$ which is obtained by applying all the rules in $I$.
        It is idempotent: $I(I(D))=I(D)$.
\end{definition}

We note $I=\noI$ to indicate that \textit{there are no inference rules, thus no extra information may be inferred.} Formally, this means $I$ is the identity function.
 We say that a database $D$ is a \textit{saturation} of $D^{(-1)}$ by inference system $I$ if  $I(D^{(-1)}) = D$. 
 Conversely, we say that $D^{(-1)}$ is an antecedent of $D$ if $D$ is a saturation of $D^{(-1)}$. Note that the database in Example~\ref{ex:saturated-hospital-db} is indeed saturated by the inference system $I$, applying the inference rule $IR$.  \\

\noindent\textbf{Differential Privacy.}
DP~\cite{dwork2006differential} is possibly the most well-established criterion in the privacy research community. It ensures that an attacker who observes the outcome of a query cannot infer the presence or absence (and hence the value) of any particular sensitive datum in the dataset. 

\begin{definition}[$\varepsilon$-Differential privacy DP~\cite{dwork2006differential}
] Given $\varepsilon > 0$, a function $f$ defined on $\mathcal{D}$, and a distance $d$ over $\mathcal{D}$, $f$ satisfies $\varepsilon$-DP if for all  $(D_1, D_2)\in\mathcal{D}^2$ such that $d(D_1, D_2) = 1$, and for all subsets $S$ of the range of $f$, we have:
$$
\Pr[f(D_1) \in S] \leq e^{\varepsilon} \times \Pr[f(D_2) \in S] 
$$
where probability is taken over the randomness of $f$. In this case, we say that $D_1$ and $D_2$ are $\varepsilon$-indistinguishable~\cite{DBLP:conf/tcc/DworkMNS06} where $\varepsilon$ represents the privacy budget and parameterizes the protection. 
\end{definition}

\vspace{0.2cm}
\noindent\textbf{Implementing DP.} A classical method to implement a DP mechanism for numeric queries is to return a noisy answer rather than the true query result~\cite{dwork-book}. The added noise must be carefully calibrated  so that the result remains useful while still protecting individual contributions. Its amplitude depends on $\varepsilon$ and  the query's sensitivity $\Delta f$ (called $\ell_1$-sensitivity by~\cite{dwork-book}), i.e. how much it may vary among a neighborhood:

\begin{definition}[Global Sensitivity \(\Delta f\)~\cite{dwork-book}]\label{def:Sensitivity}
Given a numeric query \(f : \mathcal{D} \rightarrow \mathbb{R}\) and a distance $d$ over $\mathcal{D}$, the (global or $\ell_1$) sensitivity of \(f\) is defined as

$\Delta f = \underset{x,y}{\max} \lvert f(x) - f(y) \rvert$

for all \(x, y\) such that $d(x,y) = 1$.
\end{definition}

\vspace{0.2cm}
\noindent\textbf{DP on various database models.}
 DP is immediately applicable to any space $\mathcal{D}$ given a proper distance $d$ or simply a neighborhood definition. Classically, for relational databases, the notion of neighborhood corresponds to two databases $D_1$ and $D_2$ differing by a single tuple. When considering graph databases, two DP-models, relying on two notions of neighborhoods, prevail: k-edge-DP~\cite{estimate-degree-sequence} and node-DP~\cite{node-dp}, where two databases are neighbors if they differ by up to $k$ edges or one node and all its adjacent edges, respectively. \\
 
\noindent\textbf{Bounded and Unbounded DP.} There exist two ways to compute distances for differential privacy: unbounded and bounded  DP~\cite{NearH21}. In bounded DP, two datasets are considered neighbors if one can be obtained by editing one sensitive piece of information within the other. In unbounded DP, two datasets are neighbors if they differ by the addition or deletion of a single piece of information.

\noindent We now define \textit{paired} bounded and unbounded distances, an important new concept that we introduce to link the attacker and defender models:
 \begin{definition}[Paired bounded/unbounded distances]\label{def:PairedDistance}
Two distances $d_u$ $d_b$, unbounded and bounded, respectively, defined on $\mathcal{D}$ are said to be paired if :
 
\begin{align*}
\forall (D_1, D_2) \in \mathcal{D}, \quad
d_b(D_1,D_2) = 1  \iff \exists & D_0 ,  D_0 \subseteq D_1 \wedge d_u(D_0, D_1) = 1 \\
&\quad  \wedge D_0 \subseteq D_2 \wedge d_u(D_0,D_2) = 1
\end{align*}

 \end{definition}

\section{Problem formalization and analysis}
\label{sec:pbstate}

In this paper, we consider a curator (defender) relying on DP where the true database $D$ is $\varepsilon$-indistinguishable from its neighbors according to some distance $d$.  On the other hand, the attacker tries to determine the true database among a set of databases it considers possible. 

Our attackers generalize classical worst-case-scenario attackers that are just one datum (e.g. tuple, edge, or line) away from knowing the full database, to a setting where fact inference is possible. This means that the attacker is considering \textit{unbounded} distance neighbors where exactly one piece of information $\iota_{0}$ (i.e. the missing datum and all other derived from it) is added to its prior knowledge. On the other hand, the curator (who knows $D$) must protect against \textit{any} such attacker (for all possible values of $\iota_{0}$). The  $\varepsilon$-indistinguishable databases should hence be the union of all the databases considered by \textit{any} attacker, thus the union of all the databases containing all information of $D$ except for one plus all possible variations of the missing information. This is by construction \textit{exactly} the set of databases that are \textit{bounded} neighbors of $D$. We will thus consider any paired bounded/unbounded distances $d_b$ and $d_u$, with $d_b$ the distance used by the \textit{curator} (adding the noise) and $d_u$ used to characterize the attacker.

In this section, we formalize the attackers, the defense and attack spaces, and the concept of well-suitedness. By using these definitions, we are able to show that in the classical case with no consideration of inferences, existing DP models are well suited for our attackers. \textit{An interesting problem, studied in the rest of the paper, arises when the attacker leverages knowledge of  inference rules}. In this case, classical DP models are no longer well-suited. The main result of this paper is to restore this property of well-suitedness with an extension of classical DP models that applies for various distances accounting for ontologies. \\

\noindent\textbf{Defense space.} A defense space is a mapping from $\mathcal D$ to $2^{(\mathcal{D})}$
that maps each database to a set of decoys. In DP, this set of decoys (and their corresponding defense space) is the neighborhood of the true database according to the chosen distance (plus the original database itself).
\begin{definition}[Defense space]\label{def:DefenseSpace}
    Let $d_b$ be a (bounded) distance over databases. The defense space $N_{d_b}$ of $d_b$ maps each database $D$ to $ \{D' | d_b(D,D') \leq 1\}$.
\end{definition}

\vspace{0.2cm}
\noindent\textbf{Attacker model.}
An attacker is an observer with partial knowledge of the database (i.e., a subset of data known to be true). This knowledge induces a set of databases that are consistent with their information and therefore plausible as the true database state, which we call \textit{attack space}. 
We consider a \textit{worst-case} (i.e. very knowledgeable) attacker that has knowledge of the database \textit{up to exactly one} missing datum, as in~\cite{balle2022reconstructing,lee2011much}. Such an attacker (that we term \textit{up-to-one} attacker) and its attack space are defined as follows: 

\begin{definition}[Up-to-one attacker and its attack space]\label{def:Attack}
    Let $d_u$ be an (unbounded) distance on $\mathcal D$,
    $I$ be an inference system, and
    $D_0$ be a database modeling the prior knowledge of the attacker.
    We note $A_{d_u}^{I}(D_0)$ the up-to-one attacker on distance $d_u$, 
    aware of inference system $I$, and of prior $D_0$. $A_{d_u}^{I}(D_0)$
    considers all  saturated graphs $D' \in \mathcal{D}$ such that  $\exists D'^{(-1)}$ such that
        \begin{itemize}
            \item[$\bullet$] $D'=I(D'^{(-1)})$
            \item[$\bullet$] $D_0 \subseteq D'^{(-1)}$
            \item[$\bullet$] $d_u(D_0,D'^{(-1)}) = 1$ 
        \end{itemize}
\end{definition}

Intuitively, such an attacker has a prior knowledge $D_0$ and lacks exactly one datum of information, along with any data that can be directly inferred from it. The attacker considers all databases $D'^{(-1)}$ obtained by augmenting $D_0$ with exactly one possible instantiation of the missing datum (i.e. $d_u(D_0,D'^{(-1)}) = 1$) and saturates them using the inference system $I$.  
By abuse of notation, $A_{d_u}^{I}(D_0)$ denotes both an attacker and its attack space and $A_{d_u}^{I}$ denotes both the class of attackers and the union of their attack spaces. \\

\noindent\textbf{Well-suited defense with regard to a class of attacker.} A defense space is appropriately calibrated against a class of attackers if the set of graphs they consider plausible is exactly equal to the defense space. 
If we consider DP, and the up-to-one attackers, we get the following formalization:
\begin{definition}[Well-suited DP defense] \label{def:WellSuited}
    Let $d_u, d_b$ be two distances on $\mathcal{D}$ and $I$ be an inference system. 
    We say that $d_b$-DP is a well-suited defense to  the $A_{d_u}^{I}$ up-to-one class of attackers if,
    for all $D\in\mathcal D$,
    $$N_{d_b}(D)=
    \bigcup\limits_{D_0|D\in A_{d_u}^{I}(D_0)}(A_{d_u}^{I}(D_0))$$
\end{definition}

We note that this property is indeed respected for classical DP and distances in the absence of inference rules.
\begin{lemma}[Classical DP models are well-suited w.r.t. semantic-unaware attackers]\label{prop:edgeWellSuited}
    For any $d_u$ and $d_b$ paired distances, $d_b$-DP is a well-suited defense to the $A_{d_u}^{\noI}$ up-to-one class of attackers.
\end{lemma}

\subsubsection*{Objectives and contributions of the rest of the paper.}
In the next part of the paper, we analyze the ill-suitedness (and related ill effects) of classical DP models w.r.t. semantic-aware attackers and propose $(I,d_b)$-DP that integrates ontology to a distance $d_b$, such that:
\begin{enumerate}
    \item $(\noI, d_b)$-DP is equivalent to $d_b$-DP and
    \item if  $d_b,d_u$ are paired with $d_b$ bounded, then $(I,d_b)$-DP is well-suited w.r.t. $A_{d_u}^{I}$.
\end{enumerate}

\section{Mismatch of the $d_b$-DP defense for a semantic-aware attacker $A_{d_u}^{I}$}
\label{sec:fullysemantic}
We now study the impact of the up-to-one adversary knowing inference rules  and show that classical distance are ill-suited against such an attacker.

A curator using $d_b$-DP associates to each database $D$ the defense space $N_{d_b}(D) = \{ D' | d_b(D,D') \leq 1 \}$. Given a database $D$, by Def~\ref{def:Attack}, an instance of $A_{d_u}^{I}$ considering $D$ plausible starts with a prior $D_0$ such that $\exists D^{(-1)}, D_0 \subseteq D^{(-1)} \wedge d_u(D^{(-1)},D_0) = 1 \wedge I(D^{(-1)}) = D$. It considers all other databases $D'$ with an antecedent $D'^{(-1)}$ such that $d_u(D'^{(-1)},D_0) = 1$ and $D_0 \subseteq D'^{(-1)}$. Hence, by Def.~\ref{def:PairedDistance}, it considers databases $D'$ with an antecedent $D'^{(-1)}$ such that $d_b(D^{(-1)},D'^{(-1)}) = 1$. The problem arises here since it is possible to have $d_b(D,D') > 1 $ even though $d_b(D^{(-1)},D'^{(-1)}) = 1$. Databases an attacker with inference rules considers are not necessarily neighbors in the sense of the considered distances $d_u$ or $d_b$. Indeed, inferences on different $d_b$-neighboring databases can create (arbitrarily) distant databases w.r.t. $d_b$, since an arbitrary number of facts could be added during the saturation process (this obviously depends on $I$ and $D$).

Since the defense space of the curator is entirely composed of $d_b$ neighbors of $D$, it immediately follows that there may be a mismatch and that the $d_b$-DP defense cannot always be well-suited for a semantic-aware attacker $A_{d_u}^{I}$. This is particularly problematic since the perceived query variation of an up-to-one attacker may be \textit{greater} than the curator's considered sensitivity. For any query $Q$ applicable on $\mathcal{D}$, we note $\Delta_{I} Q$  the perceived sensitivity of $A_{d_u}^{I}$ attackers, i.e. 

$$\Delta_{I} Q = \max_{D_0 }( \max_{(D,D') \in (A_{d_u}^{I}(D_0))^2} |Q(D)-Q(D')|)$$

\begin{proposition}[Privacy leakage by a $d_b$-DP curator against an $A_{d_u}^{I}$ attacker]\label{prop:semanticMismatch}
Consider $A_{d_u}^{I}$, a class of attackers with knowledge of 
inference rules on a database, and a curator using $d_b$-DP, thus a defense space not considering inference rules. It is possible for such a curator to underestimate the leakage of sensitive private information; i.e. that for a query $Q$, $\Delta_{\noI} Q < \Delta_{I} Q$. In fact, it is possible that $\Delta_{\noI} Q = 0 $ while $\Delta_{I} Q>0$.
\end{proposition}

\begin{proof}
We illustrate this mismatch in an example inspired by~\cite{al2024patient}. We consider edge-distances and $\mathcal{D}$ to be KGs such as the one in Fig.~\ref{fig:DBex} representing de-identified data in a hospital. In such graphs, each \( d_i \) is \textit{necessarily} assigned to a  \( dept_j \), and each \( p_k \)  is \textit{necessarily} a patient of a specific \( d_i \).
All of them are instances of a class \texttt{dept}, \texttt{doctor}, or \texttt{patient}, which are represented as nodes. Typing is represented as an edge.  Edges are labeled \textit{hasPatient}, \textit{worksIn}, \textit{hasType}, or \textit{patientIn}. Contrarily to aformentioned relations, the \textit{patientIn} relation is \textit{not mandatory} and links the \texttt{patient} of a \texttt{doctor} to the \texttt{departement} it belongs to.

We consider the inference rule stated in Example~\ref{ex:saturated-hospital-db} that a \texttt{patient} under the care of a \texttt{doctor} working in a particular \texttt{dept} is a \texttt{patient} in said \texttt{dept}.

 \begin{figure}[h]
    \centering
    \begin{tikzpicture}[->, node distance=2cm, scale=0.6, transform shape,  
        main/.style={circle, draw, fill=white, minimum size=1.5cm, inner sep=0pt, align=center, font=\small}]

        \node[main] (dept1) {dept$_1$};
        \node[main, right of=dept1,xshift=15pt] (dept2) {dept$_2$};
        \node[main, below of=dept1] (smith) {$d_1$};
        \node[main, right of=smith,xshift=15pt] (adam) {$d_2$};
        \node[main, below of=smith] (p1) {$p_1$};
        \node[main, right of=p1] (p2) {$p_2$};
        \node[main, right of=p2,xshift=15pt] (p3) {$p_3$};

        \node[main, right of=p3,xshift=15pt] (c3) {patient};
        \node[main, above of=c3] (c2) {doctor};
        \node[main, above of=c2] (c1) {dept};

        \draw[->] (adam) -- node[midway, right] {worksIn} (dept2);
        \draw[->] (smith) -- node[midway, left] {worksIn} (dept1);
        \draw[->] (smith) -- node[midway, left] {hasPatient} (p1);
        \draw[->] (smith) -- node[midway, right] {hasPatient} (p2);
        \draw[->] (adam) -- node[midway, right] {hasPatient} (p3);

        \begin{scope}[on background layer]
        
        \draw[->] (dept2) -- node [near end, above] {hasType} (c1);
        \draw[->] (dept1) to [out=-10, in=190] (c1);
        
        \draw[->] (adam) -- node [near end, above] {hasType} (c2);
        \draw[->] (smith) to [out=-10, in=190] (c2);
        
        \draw[->] (p3) -- (c3);
        \draw[->] (p2) to [out=-10, in=190] (c3);
        \draw[->] (p1) to [out=-20, in=200]  node [pos=.9, below] {hasType} (c3);

        \end{scope}

    \end{tikzpicture}
    \caption{Considered databases $\mathcal{D}$}
    \label{fig:DBex}
\end{figure}
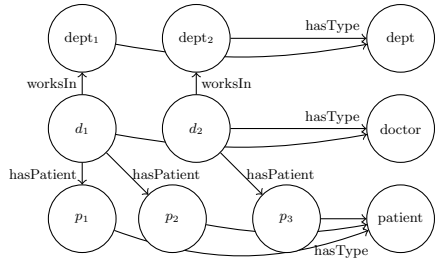

\vspace{0.2cm}

We now consider an example database $D^{ex}$ and its saturation by $I$, as represented in Figure~\ref{fig:graph_saturation}, where types are implicit for simplicity sake. We consider the basic query $Q$ which counts the number of patients in the \texttt{oncology} department. 
\begin{lstlisting}
Q = SELECT (COUNT(DISTINCT ?patient) AS ?numPatients)
WHERE {
  ?patient :patientIn :Oncology .
}
\end{lstlisting}

We further restrict $\mathcal{D}$ to only contain saturated database, considering that the curator systematically saturates the database as Q would e.g. returns 0 on $D^{ex}$. 

If an attacker can start from a prior $D_0$, guess a database $D'$ with one more piece of data (\textit{i.e.} one more edge), then saturate that $D'$ into a database $I(D') \in \mathcal{D}$, then $I(D')$ will be considered by the attacker as part of the attack space. \textit{This means that an attacker considers a specific saturated database if its prior is a subgraph (unbounded) neighbor of one of this database's antecedents}, i.e. a graph that saturates to it. In order to provide adequate protection, the curator should also consider them as part of the defense space.

In Fig.~\ref{fig:TrimmedPrior}, database $D_0^{ex}$ represents the attacker's knowledge. Note that in this case, $D_0^{ex} \notin \mathcal{D}$ since there exists a \textit{doctor} that does not work in any \textit{department}.  

  \begin{figure}[h]
    \centering
    \begin{subfigure}{0.45\textwidth}
    \centering
    \begin{tikzpicture}[->, node distance=2cm, scale=0.6, transform shape,  
        main/.style={circle, draw, minimum size=1.5cm, inner sep=0pt, align=center, font=\small}]

        \node[main] (dept1) {Psychiatry};
        \node[main, right of=dept1,xshift=15pt] (dept2) {Oncology};
        \node[main, below of=dept1] (smith) {DrSmith};
        \node[main, right of=smith,xshift=15pt] (adam) {DrAdam};
        \node[main, below of=smith] (p1) {P1};
        \node[main, right of=p1] (p2) {P2};
        \node[main, right of=p2,xshift=15pt] (p3) {P3};

 consider       \draw[->] (adam) -- node[midway, right] {worksIn} (dept2);
        \draw[->] (smith) -- node[midway, left] {worksIn} (dept1);
        \draw[->] (smith) -- node[midway, left] {hasPatient} (p1);
        \draw[->] (smith) -- node[midway, right] {hasPatient} (p2);
        \draw[->] (adam) -- node[midway, right] {hasPatient} (p3);

    \end{tikzpicture}
    \caption{Antecedent of the true database $D^{ex}$}
    \label{fig:TrimmedTrue}
    \end{subfigure}
    \hfill 
    \begin{subfigure}{0.45\textwidth}
      \centering
\begin{tikzpicture}[->, node distance=2cm,scale=0.6, transform shape,  main/.style={circle, draw, minimum size=1.5cm, inner sep=0pt, align=center, font=\small}]
 
    \node[main] (dept1) {Psychiatry};
    \node[main, right of=dept1,xshift=50pt] (dept2) {Oncology};
    \node[main, below of=dept1] (smith) {DrSmith};
    \node[main, right of=smith,xshift=50pt] (adam) {DrAdam};
    \node[main, below of=smith] (p1) {P1};
    \node[main, right of=p1] (p2) {P2};
    \node[main, right of=p2] (p3) {P3};

    \draw[->] (adam) -- node[midway, right] {worksIn} (dept2);
    \draw[->] (smith) -- node[midway, right] {worksIn} (dept1);

    \draw[->] (smith) -- node[midway, left] {hasPatient} (p1);
    \draw[->] (smith) -- node[midway, right] {hasPatient} (p2);
    \draw[->] (adam) -- node[midway, right] {hasPatient} (p3);
    \draw[->, xshift=10pt, draw=blue] (p1) to[out=180, in=180] node[midway, text=blue] {patientIn} (dept1);
    \draw[->, xshift=7pt, draw=blue] (p2) to[out=50, in=0] node[midway, text=blue] {patientIn} (dept1);
    \draw[->, xshift=10pt, draw=blue] (p3) to[out=50, in=0] node[midway, text=blue] {patientIn} (dept2);
\end{tikzpicture}
    \caption{True database $I(D^{ex})$ }
     \label{fig:Saturated}
    \end{subfigure}
     \caption{True database and inferred information}
    \label{fig:graph_saturation}
\end{figure}
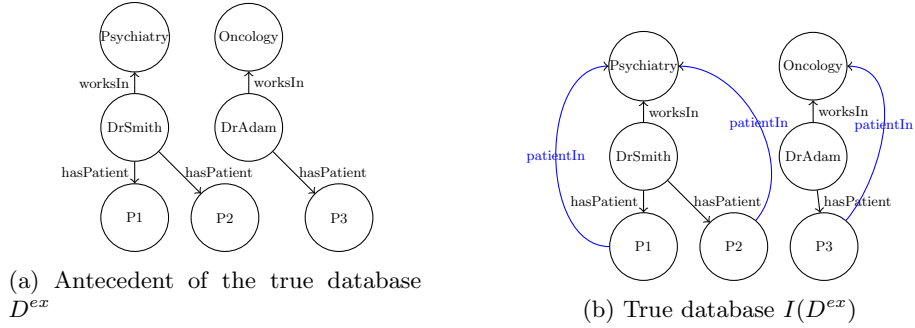

 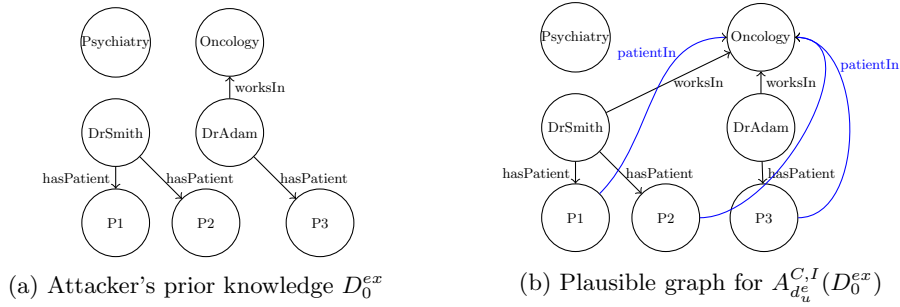
\begin{figure}[h]
    \centering

    \begin{subfigure}{0.45\textwidth} 
    \centering
    \begin{tikzpicture}[->, node distance=2cm,scale=0.6, transform shape,  main/.style={circle, draw, minimum size=1.5cm, inner sep=0pt, align=center, font=\small}]
    \node[main] (dept1) {Psychiatry};
    \node[main, right of=dept1,xshift=15pt] (dept2) {Oncology};
    \node[main, below of=dept1] (smith) {DrSmith};
    \node[main, right of=smith,xshift=15pt] (adam) {DrAdam};
    \node[main, below of=smith] (p1) {P1};
    \node[main, right of=p1] (p2) {P2};
    \node[main, right of=p2,xshift=15pt] (p3) {P3};

    \draw[->] (adam) -- node[midway, right] {worksIn} (dept2);

    \draw[->] (smith) -- node[midway, left] {hasPatient} (p1);
    \draw[->] (smith) -- node[midway, right] {hasPatient} (p2);
    \draw[->] (adam) -- node[midway, right] {hasPatient} (p3);

    \end{tikzpicture}
    \caption{Attacker's prior knowledge $D_0^{ex}$}
     \label{fig:TrimmedPrior}
    \end{subfigure}
    \hfill 
    \begin{subfigure}{0.45\textwidth} 
    \centering
\begin{tikzpicture}[->, node distance=2cm,scale=0.6, transform shape,  main/.style={circle, draw, minimum size=1.5cm, inner sep=0pt, align=center, font=\small}]

    \node[main] (dept1) {Psychiatry};
    \node[main, right of=dept1,xshift=60pt] (dept2) {Oncology};
    \node[main, below of=dept1] (smith) {DrSmith};
    \node[main, right of=smith,xshift=60pt] (adam) {DrAdam};
    \node[main, below of=smith] (p1) {P1};
    \node[main, right of=p1] (p2) {P2};
    \node[main, right of=p2,xshift=5pt] (p3) {P3};

    \draw[->] (adam) -- node[midway, right] {worksIn} (dept2);
    \draw[->] (smith) -- node[midway, right] {worksIn} (dept2);

    \draw[->] (smith) -- node[midway, left] {hasPatient} (p1);
    \draw[->] (smith) -- node[midway, right] {hasPatient} (p2);
    \draw[->] (adam) -- node[midway, right] {hasPatient} (p3);
    \draw[->, xshift=10pt, draw=blue] (p1) to[out=45, in=180] node[near end, left, text=blue] {patientIn} (dept2);
    \draw[->, xshift=10pt, draw=blue] (p2) to[out=0, in=0] (dept2);
    \draw[->, xshift=10pt, draw=blue] (p3) to[out=0, in=0] node[near end, right, text=blue] {patientIn} (dept2);
\end{tikzpicture}
    \caption{Plausible graph for $A_{d_u^e}^{C, I}(D_0^{ex})$ }
     \label{fig:SaturatedNeighbour}
    \end{subfigure}
     \caption{Example of attacker knowledge and saturation of a possible neighbor}
    \label{fig:graph_comparison} 
\end{figure}

The attacker knows  the  rule 
I = \{\texttt{hasPatient(?x, ?y)} $\land$  \texttt{worksIn(?x, ?z)} $\Rightarrow$ \texttt{patientIn(?y, ?z)} \}, considers the neighbors of $D_0^{ex}$, then saturates them. 
Fig.~\ref{fig:Saturated} and Fig.~\ref{fig:SaturatedNeighbour} show 
two plausible saturated graphs for this attacker. This means that \textit{by construction} they differ by a single datum of information for the attacker. Hence, $\Delta_I Q \geq 2$. However, from the curator's point of view, if they were to use $d_b$-DP as a defense without further adjustment, the database pictured in Fig.~\ref{fig:SaturatedNeighbour}) would not be part of the defense space of the database $I(D_{0}^{ex})$ (Fig.~\ref{fig:Saturated}), as they are too far from each other (the edge distance between these graphs is 3, not 1). This means that these graphs will not be considered by the curator when computing the sensitivity of the query. 

Indeed, the direct $d_b$-distance neighbors of the database $I(D^{ex})$ (and actually many saturated graphs) will have the same answer to our query, as any alteration to a \texttt{worksIn} or \texttt{hasPatient} edge may have a knockdown effect on \texttt{patientIn} and would result in the database no longer be saturated (and hence $\notin \mathcal{D})$.

In fact, two saturated graphs can only be $d_b$-distance neighbors if their only difference
is swapping a patient between two doctors of same department, or changing the department of a patientless doctor.
Under those circumstances, as stated in the proposition, we find a sensitivity equal to 0 under $d_b$-DP on $\mathcal{D}$. 

The code and demo of our tool illustrating this possible mismatch are available at: \href{https://github.com/Yasmine-Hayder/Onto-Differential-Privacy}{https://github.com/Yasmine-Hayder/Onto-Differential-Privacy}

\end{proof}

\noindent\textbf{Consequence of this mismatch.} In this case, this ill-suitedness has \textit{drastic} consequences in particular when restricting the considered space to saturated databases:  following the formal definition of DP, a sensitivity equal to 0  would mean that a query with no added randomized mechanism (i.e. outputing the true result of the query) is $\varepsilon$ differentially private for all $\varepsilon > 0$. Indeed, since the result of the query does not vary among any neighborhood, the probability of obtaining a result is the same when applied to any adjacent databases.

However \textit{in reality}, the query's raw results \textit{do} reveal some sensitive information to our up-to-one attacker, which could derive its missing piece of information \textit{with absolute certainty}. The mismatch of attack and defense spaces would thus here lead to a dramatic \textbf{complete absence of protection under classical DP}.

\section{$(I,d_b)$-DP : A defense model against a semantic-aware attacker $A_{d_u}^{I}$}
\label{sec:solution}

In order to solve the sensitivity issue when using traditional DP when dealing with ontology-aware attackers,
we introduce onto-DP, or $(I,d_b)$-DP, based on a novel distance notion that contains a reasoning process.
This distance is constructed to have a defense space matching semantic-aware attackers. As such, it considers saturated graphs, and neighbors in this distance are not neighbors in the sense of $d_b$,  but rather have antecedents that are $d_b$ neighbors. Hence, they are also unbounded neighbors of a common graph, matching with some attacker's prior. This is illustrated in Fig.~\ref{fig:FullDistance} and formalized as follows.

\begin{definition}[$(I,d_b)$-onto Distance]
Let $d_b$ be a (bounded) distance.
The $(I, d_b)$-ontology aware  distance is defined by its neighborhoods, then classically extended. For $(D,D')\in \mathcal{D}^2$, $D'$ is a neighbor of $D$ if and only if there exists an antecedent $D^{(-1)}$ of $D$ and $D'^{(-1)}$ of $D'$ w.r.t. $I$ such that $d_b(D^{(-1)}, D'^{(-1)}) = 1$.

\end{definition}

\noindent \textbf{Consequence.} It is immediate that $(\noI, d_b)$-ontology aware distance is $d_b$.
\begin{figure}[h] 
    \centering
\begin{tikzpicture}[->, node distance=2cm,scale=0.7, transform shape,  main/.style={draw, minimum size=1.5cm, inner sep=0pt, align=center, font=\small}]
    \node[main] (gbot) {$D_0$};
    \node[main, above left of=gbot,xshift=-30pt] (g0) {$D^{(-1)}$};
    \node[main, above right of=gbot,xshift=30pt] (g1) {$D'^{(-1)}$};
    \node[main, above of=g0, xshift=-30pt, yshift=20pt] (g) {$D$};
    \node[main, above of=g1, xshift=30pt, yshift=20pt] (g') {$D'$};

    \begin{scope}[on background layer]
    \draw[color=gray, fill=gray, opacity=0.25] ($(g0)+(-1,1)$) rectangle ($(g1)-0.5*(g1)-0.5*(g0)+(1,-2)$);
    \node[below of=gbot, yshift=10pt] (unsat) {Unsaturated Databases};
    \end{scope}
1
    \draw[<->] (gbot) to[out=180, in=270] node[midway, below] {unbounded} (g0);
    \draw[<->] (gbot) to[out=0, in=270] node[midway, below] {unbounded} (g1);
    \draw[<->] (g0) to node[midway, above] {bounded} (g1);

    \draw[->, dashed] (g0) to[out=90, in=270] node[midway, below] {saturate} (g);
    \draw[->, dashed] (g1) to[out=90, in=270] node[midway, below] {saturate} (g');
    \draw[<->] (g) to node[midway, above] {onto neighbors} (g');

\end{tikzpicture}
    \caption{$(I,d)$ neighborhood pattern}
     \label{fig:FullDistance}
\end{figure}
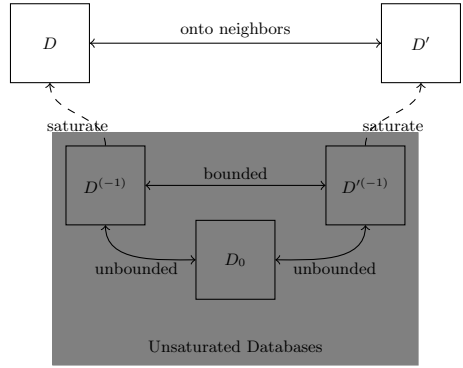

\begin{theorem}[Onto-DP is well-suited w.r.t. semantic aware attacker]\label{prop:FullWellSuited}
    Let $d_u,d_b$ be paired distances, unbounded and bounded respectively.
    For any $I$ inference rules,
    ($I, d_b$)-DP is a well-suited defense to the $A_{d_u}^{I}$ up-to-one class of attackers

\end{theorem}

\begin{proof}
Let us consider a database $D\in I(\mathcal{D})$,
and a $A_{d_u}^{I}$ up-to-one attacker of prior $D_0$ that considers it. This means that there exists a database $D^{(-1)}$, neighbor of $D_0$ according to $d_u$,
such that $I(D^{(-1)})=D$.
Any other databases considered by this attacker are $D'\in I(\mathcal{D})$ such that
there exists a database $D'^{(-1)}$, $d_u$ neighbor of $D_0$, whose saturation $I(D'^{(-1)})$ equals to $D'$.
Since $D^{(-1)}$ and $D'^{(-1)}$ are both $d_u$ neighbors of the same sub-database $D_0$,
they are $d_b$ neighbors by definition of paired distance. By definition, $D$ and $D'$ are neighbors in the $(I,d_u)$-distance.

Conversely, let us consider a couple of databases $(D,D')\in I(\mathcal{D})^2$, neighbors in the $(I, d_b)$-distance.
By definition, there exist $D^{(-1)}$ antecedent of $D$ and $D'^{(-1)}$ of $D'$
that are $d_b$ neighbors.
Since $d_b$ and $d_u$ are paired, this means that there exists some $D_0$,
subset and $d_u$-neighbor of both $D^{(-1)}$ and $D'^{(-1)}$.

This means that a $(I,d_u)$ up-to-one attacker of prior $D_0$ considers both $D$ and $D'$.

From both of these we conclude that the defense space of $D$ in the $(I, d_b)$-onto-DP is exactly the set of all databases considered by at least one $A_{d_u}^{I}$ up-to-one attacker that also considers $D$.
\end{proof}
    
\noindent\textbf{Takeaway.} Theorem~\ref{prop:FullWellSuited} is the main result of the article. It shows that it is possible to build a DP mechanism that will correctly evaluate the sensitivity of queries, in the presence of attackers that have knowledge of  inference rules on a database, by using our proposed $(I,d_b)$-distance.

\section{Related Work}
\label{sec:rw}

Since KGs are the traditional representation for knowledge centered databases, we provide herein an overview of existing work related to semantic-aware DP in such context. We also discuss work questioning attacks models for DP and investigating distances and neighborhoods for DP. Finally, we present the proposal closest to our own, DP approaches over correlated data. \\

\noindent\textbf{DP for KGs.} Even though KGs are the traditional representation for knowledge centered databases, there is surprisingly little work proposing semantic-aware DP approaches on KG. Standard DP approaches are oftentimes applied, for example,~\cite{hu2023quantifying} applies ``triple''-DP, which is equivalent to traditional edge-DP, to the problem of federated KG embedding.  \\

Reuben~\cite{reuben2018towards} was, to the best of our knowledge, the first to propose semantic-aware DP for edge-labeled directed graphs. This approach was limited to applying edge-DP on a subset of labels. Building on this work, Taki et al.~\cite{10.1007/978-3-031-15743-1_20} proposed a projection-based approach to reduce the sensitivity of queries on KG using QL-edge-DP.  Han et al.~\cite{han2022framework} proposed a similar idea where a set of sensitive relationships is specified. 

\noindent\textit{To the best of our knowledge, semantic-aware DP approaches for KG are limited to the consideration of sensitive and non-sensitive labels (or types), and there exists no work integrating constraints and inferences in DP approaches.} \\

\noindent\textbf{Attacks models on DP.} The efficiency of the protection of DP is reliant on the choice of an $\varepsilon$. However, the concrete guarantees such constraints provide is heavily reliant on the type of attack scenario considered.

In the reconstruction attacks~\cite{balle2022reconstructing},
one standard scenario is the \emph{informed attacker},
who knows everything about a model and its training data save for one specific element.
Our worst-case scenario is inspired by the implicit threat model of DP~\cite{DBLP:conf/sp/NasrSTPC21}.
Indeed, previous works~\cite{lee2011much} use worst-case attack scenarii
to measure the efficiency of $\varepsilon$-DP. Their illustrative choice is that
of an attacker that tries to identify members of a queriable subgroup of individuals,
but only misses one information to do so, a scenario which serves as an inspiration for our up-to-one attacker. Some works like \cite{soria2017individual,o2019bootstrap} explored relaxations of this worst-case attacker, but \cite{protivash2022reconstruction} demonstrates that such relaxations can lead to vulnerabilities.

\vspace{0.2cm}
\noindent\textbf{Questioning DP-distances and neighborhoods.}
Several existing works also study the necessary departures from classic distances and mechanisms
in scenarii where the neighborhoods given by classic distances are of variable believability or usefulness. 
In geo-indistinguishability~\cite{chatzikokolakis2015geo} (and more generally metric-DP~\cite{chatzikokolakis2013broadening}), which aims to privately publish someone's location,
some locations (e.g. a river, the sea) are considered unlikely answers that cannot realistically count as convincing decoys. \\

\noindent\textbf{DP over correlated data.} In many real cases like social networks, data is related and cannot be considered independent~\cite{liben2003link}. To deal with this, models like Pufferfish \cite{kifer2014pufferfish} and Bayesian Differential Privacy (BDP) \cite{yang2015bayesian} were developed.
These formalisms provide a very wide setting to customize privacy for dependent data by defining secrets, \textit{i.e.} alternative worlds that must be indistinguishable, and an attacker’s background knowledge.
Under such a setting, DP can be defined as requiring any attacker,
regardless of prior knowledge, to be unable to gain intel from an output.
In such a setting, one would model deterministic inferences as a restriction on attackers' background knowledge: an informed attacker would only consider databases where all inferences have been made.
While this accurately translates the definition of DP,
and lets the authors show that under any correlation, classical Laplace mechanisms allow for security leaks,
one of the keys to DP's ease of use is that
rather than checking for all attackers whether they can gain undue confidence on a specific secret,
one can focus on showing it for worst-case scenario attackers, that only lack one element of a database, and extending the property to the general case.
To find proper background knowledge for attackers that are able to use inference to deduce several facts from one guess is then necessary to have a notion of semantically aware DP one can hope to demonstrate on a given process.

\section{Conclusion}
\label{sec:ccl}
This paper explores the challenges of protecting sensitive information
against attackers with knowledge about the database semantics (i.e. data dependency). 
We introduced and investigated \textit{semantic-aware} attackers, who have knowledge of the inference rules associated with the database, and showed that traditional DP methods may underestimate the knowledge gained by such attackers, leading to privacy leaks. 
To address these issues, we proposed \textit{onto-DP}, an extension of existing differential privacy paradigms that enrich them with the consideration of inference rules.  \\

 We believe these results open exciting new research directions at the intersection of DP and semantically rich databases such as KGs. 

A natural next step is a systematic empirical evaluation of these semantic neighborhoods under realistic rule fragments and practical reasoners. In particular, we plan to quantify how inference impacts neighborhood size and global sensitivity, and to design concrete algorithms for efficiently computing neighborhoods under \textit{onto-DP}. Such an evaluation is essential to understand whether semantic-aware privacy guarantees remain practically attainable, or whether inference causes prohibitive sensitivity blow-ups.

Another promising direction is to extend our analysis to non-numeric queries, for example via the Exponential Mechanism\cite{DBLP:conf/focs/McSherryT07}, which is well-suited for categorical outputs. Since this mechanism relies on a notion of distance between plausible answers, understanding how semantic neighborhoods reshape this notion under inference rules becomes particularly important.

\begin{credits}
\subsubsection{\ackname}  This work was supported by grants ANR-22-PECY-0002 (IPoP project),
ANR-23-CMAS-0019 (CyberINSA project), and ANR-23-CE23-0032 (DIFPRIPOS project) funded by the ANR.

\subsubsection{\discintname}
The authors have no competing interests to declare that are
relevant to the content of this article. 
\end{credits}

\bibliographystyle{unsrt}
\bibliography{sample-base}

\appendix

\end{document}